\documentclass[a4paper,8pt]{article}
\usepackage{microtype}
\usepackage{amssymb}
\usepackage{amsmath}
\usepackage{paralist}
\usepackage{tikz}
\usepackage{graphicx}
\usepackage{hyperref}
\usepackage{subcaption}
\usepackage{cite}
\usepackage{complexity}
\usepackage{authblk}
\usepackage{amssymb,amsmath,amsthm, amsfonts}

\graphicspath{{./graphics/}}

\DeclareMathOperator{\bd}{bd}
\newlang{\Parity}{Parity}
\newlang{\BGGM}{BGGM}
\newclass{\quasiNC}{quasiNC}
\newclass{\Lhard}{L-hard}
\newclass{\BWBP}{BWBP}
\DeclareMathOperator{\AND}{AND}
\DeclareMathOperator{\NOT}{NOT}
\DeclareMathOperator{\OR}{OR}
\DeclareMathOperator{\XOR}{XOR}

\newtheorem{theorem}{Theorem}
\newtheorem{lemma}{Theorem}

\title{Circuit Complexity of Bounded Planar Cutwidth Graph Matching}

\author[1]{Aayush Ojha}
\affil[1]{Indian Institute of Technology, Kanpur \authorcr
  Email: aayushoj@cse.iitk.ac.in}
\author[2]{Raghunath Tewari}
\affil[2]{Indian Institute of Technology, Kanpur \authorcr
  Email: rtewari@cse.iitk.ac.in}

\begin{document}

\maketitle

\begin{abstract}

Recently, perfect matching in bounded planar cutwidth bipartite graphs (\BGGM) was shown to be in $\ACC^0$ by Hansen et al. \cite{Hansen2014}. They also conjectured that the problem is in $\AC^0$. 

In this paper, we disprove their conjecture by showing that the problem is not in $\AC^0[p^{\alpha}]$ for every prime $p$. Our results show that the previous upper bound is almost tight. Our techniques involve giving a reduction from  {\Parity} to {\BGGM}.  A further improvement in lower bounds is difficult since we do not have an algebraic characterization for $\AC^0[m]$ where $m$ is not a prime power. Moreover, this will also imply a separation of $\AC^0[m]$ from $\P$. Our results also imply a better lower bound for perfect matching in general bounded planar cutwidth graphs. 

 \end{abstract}

\section{Introduction}
For a graph $G=(V,E)$ a {\em matching} $M \subseteq E$ is a set of edges in $G$ such that no two edges in $M$  share a common vertex. We say $G$ has a {\em perfect matching} if there exists a matching that matches every vertex in $G$. Since every graph is not guaranteed to have a perfect matching, computing a matching of maximum cardinality is a natural generalization of the perfect matching problem. The computational complexity of the matching problem is a well-studied problem particularly in the context of circuit complexity and derandomization. 

In 1965, Edmonds showed that computing maximum matching is in {\P} \cite{Edmonds65}. In 1979, Lov\'{a}sz gave an efficient randomized parallel algorithm for the perfect matching problem by showing that it is in {\RNC} \cite{Lovasz79}. The construction version of the problem was also shown to be in {\RNC} \cite{MatchingRandomNC, MulmuleyRandomizedNC}. It is an important open question whether matching has an efficient deterministic parallel
algorithm, that is, whether it is in {\NC}. Attempts to derandomize the above approaches has proved elusive so far. Recently there has been some progress on this problem. Perfect matching was shown to be in {\quasiNC} for bipartite graphs \cite{BipartiteQuasiNC} and in a subsequent paper extended to general graphs \cite{GeneralQuasiNC}.


Stronger results are known for perfect matching in graphs with bounded treewidth and its subclasses. Elberfeld et al. \cite{Tantau-LogSpace} showed that the problem is in {\L} for graphs with bounded treewidth by proving the logspace versions of Bodlaender's and Courcelle's theorem. This gives a tight bound on the complexity of perfect matching in bounded treewidth graphs since it was already known to be {\L}-hard \cite{TreeWidthMatching-L-complete}. In a subsequent paper, Elberfeld et al. showed that given a tree decomposition of the input graph as a term representation, perfect matching for bounded treewidth graphs is in uniform $\NC^1$ and for bounded tree-depth graphs is in uniform $\AC^0$ \cite{Tantau-NC1}. The upper bound of ${\NC}^1$ for bounded treewidth graphs is tight since Barrington showed that the problem is hard for ${\NC}^1$ under projection reductions \cite{Barrington-NUDFA}. 

In 2014, Hansen et al. used the characterization of Barrington and Th\'{e}rien \cite{Barrington-Therien} and showed that bipartite perfect matching in graphs with bounded planar cutwidth is in ${\ACC}^0$ \cite{Hansen2014}. They also gave a lower bound of ${\AC}^0$ for the same problem. For perfect matching in general bounded planar cutwidth graphs, they gave a lower bound of $\AND \circ \OR \circ \XOR \circ {\AC}^0$. In their paper, Hansen et al. also conjectured that perfect matching for bipartite bounded planar cutwidth graphs is in ${\AC}^0$ and for general bounded planar cutwidth graphs is in ${\AC}^0[2]$. 

\subsection{Our Result and Proof Outline}
We refute both the conjectures in this paper by giving improved lower bounds for perfect matching in bounded planar cutwidth graphs for both bipartite and general graphs. We show that perfect matching for bounded planar cutwidth graph is not in ${\AC}^{0}[p^\alpha]$ for every prime $p$ and $\alpha \in \mathbb{N}$. 

To show this improved lower bound we first reduce {\Parity} to perfect matching in bipartite bounded planar cutwidth graphs using a family of $\AC^0$ circuits. This is done by constructing certain graph gadgets as defined in Section \ref{sec:bggm2parity}. This reduction and result by Razborov \cite{Raz89} and Smolensky \cite{Smo87} shows the problem is not in ${\AC}^{0}[p^\alpha]$ for odd prime $p$. To extend the result for the case when $p=2$ we use the monoid word reduction of matching in bipartite bounded planar cutwidth graphs provided by Hansen et al. \cite{Hansen2014}. Using this reduction we show that $\Mod_q$ can be reduced to perfect matching in bipartite bounded planar cutwidth graphs for some odd prime $q$. This shows that perfect matching for bipartite bounded planar cutwidth graphs is not in ${\AC}^{0}[2^\alpha]$ as well. We also show similar lower bound for series-parallel graphs. An upper bound of ${\NC}^{1}$ for perfect matching in series-parallel graphs follows from the result of \cite{Tantau-NC1}.

\subsection{Organization of the Paper}
The rest of the paper is organized as follows. In Section \ref{sec:prelim} we will cover the preliminaries and notations that we will be using throughout the paper. We also discuss the work of Th\'{e}rien and Barrington \cite{Barrington-Therien} and results from Hansen et al. \cite{Hansen2014}. In Section \ref{sec:bggm2parity} we show the reduction of {\Parity} to perfect matching in bipartite bounded planar cutwidth graphs. In Section \ref{sec:bggmlowerbound} we first discuss the reduction framework of {\BGGM} to the monoid word problem due to Hansen et al. \cite{Hansen2014}. We then use this framework to show that perfect matching in bipartite bounded planar cutwidth graphs in not in ${\AC}^{0}[2^\alpha]$.  In Section \ref{sec:appl} we discuss an application of our result to perfect matching in series-parallel graphs. We also discuss possible limitations of our approach and future directions.

\section{Preliminaries}
\label{sec:prelim}
In this section, we give the required definitions and notations that we use in this paper. We also state the results from previous work that we use in our paper.

\subsection{Definitions and Notations}
Circuits are a non-uniform model of computation where size and depth of the circuit are two common resources that are usually studied. Additionally, type of gates used in the circuit and fan-in (indegree of a gate) are also often considered. ${\AC}^0$ is the class of problems having a family of circuits that have the constant depth, polynomial size and unbounded fan-in $\AND$, $\OR$ and $\NOT$ gates. $\AC^0[m]$ is an extension of $\AC^0$ where the circuits are allowed to have $\Mod_m$ gates in addition to $\AND$, $\OR$ and $\NOT$ gates. $\ACC^0$ is an extension of $\AC^0[m]$ where circuits are allowed to have $\Mod_m$ gates for any $m \in \mathbb{N}$. ${\NC}^1$ is the class of problems having a family of circuits that have logarithmic depth, polynomial size and bounded fan-in $\AND$, $\OR$ and $\NOT$ gates. It is easy to see that ${\AC}^0 \subseteq {\AC}^{0}[m] \subseteq {\ACC}^0 \subseteq {\NC}^1$. In fact, the first containment is proper. The reader can refer to the book by Vollmer for more details about these classes and circuit complexity in general \cite{Vol}. Let $N_a(x)$ be number of times symbol $a$ appears in string $x$. We also consider the language $\Parity = \{x \in \{0,1\}^* \mid N_1(x) \not\equiv 0 \mod 2\}$ and its generalization $\Mod_p = \{x \in \{0,1\}^* \mid N_1(x) \not\equiv 0 \mod p\}$ for any $p \geq 2$, for proving our lower bounds.

A {\em monoid} $\mathcal{M}$ is a set $S$ along with a binary operator $\oplus$ such that (i) for all $s_1,s_2 \in S$, $s_1 \oplus s_2 \in S$ (closure property), (ii) for all $s_1,s_2, s_3 \in S$ we have $s_1 \oplus (s_2 \oplus s_3) = (s_1 \oplus s_2) \oplus s_3$ (associativity property) and (iii) there exists $e \in S$ such that for all $s \in S$ we have $e \oplus s = s = s \oplus e$ (existence of identity). A subset $\mathcal{G}$ of $\mathcal{M}$ is a group in $\mathcal{M}$ if $\mathcal{G}$ is a group with respect to the operation of $\mathcal{M}$. If every group in a monoid is trivial then the monoid is said to be an {\em aperiodic monoid}. If every group in a monoid is solvable then the monoid is said to be a {\em solvable monoid}. 
For a monoid $\mathcal{M}$, the {\em monoid word problem} is given $x_1, x_2, \ldots , x_n \in \mathcal{M}$ as input, to compute $x_1\oplus x_2 \oplus \ldots \oplus x_n$. 

A {\em grid graph} is a graph $G$ embedded in an integer lattice such that each edge is either horizontal or vertical. 
A {\em grid layered planar graph} is a planar graph $G$ embedded in an integer lattice such that if there is an edge between $(a,b)$ and $(c,d)$ then $|a-c| \leq 1$. Length and width of a grid layered planar graph are the number of columns and rows in the graph respectively.

For a linear arrangement of vertices of a graph $G$, the maximum number of edges cut by any vertical line is called cutwidth of the linear arrangement of $G$. {\em Cutwidth} of $G$ is the minimum cutwidth of a linear arrangement over all possible linear arrangements of $G$. For a linear arrangement of vertices of a graph $G$  without edge crossings, the maximum number of edges cut by any vertical line is called planar cutwidth of the linear arrangement. {\em Planar cutwidth} of $G$ is the minimum planar cutwidth of a linear arrangement over all possible planar linear arrangements of $G$. If a planar linear arrangement of $G$ is not possible, for example in non-planar graphs, we define planar cutwidth of $G$ to be infinite. Note that planar cutwidth of planar graphs is not same as cutwidth of planar graphs.

Graphs with bounded planar cutwidth can be converted into grid layered planar graph preserving matching and bipartiteness. Since constructing such an embedding is not known to be in ${\NC}^1$ and supposed to be hard for ${\NC}^1$, we will assume that input is provided as a bipartite grid layered planar graph. This assumption on input is also made by Hansen et al. \cite{Hansen2014}. Here we consider the circuit complexity of perfect matching in bipartite grid layered planar graphs. Formally {\BGGM} is the set of instances of bipartite grid layered planar graphs along with their embeddings such that they have a perfect matching.  


\subsection{Algebraic Characterization of Classes in $\NC^1$}

We start by describing the definition of bounded width polynomial size programs over monoids as given in \cite{Barrington-NUDFA}. An {\em instruction} $I$ over monoid $\mathcal{M}$ is a 3-tuple $\langle j, a_0, a_1 \rangle$ where $j \in \mathbb{N}$ and $a_0, a_1 \in \mathcal{M}$. For some string $x \in \{0,1\}^{*}$, we define $I(x) = a_{x_j}$. For some $n \in \mathbb{N}$, a {\em bounded width polynomial size branching program} (in short {\BWBP}) is a tuple of polynomial number of instruction over some finite monoid $\mathcal{M}$. If $P = (I_1,I_2, \ldots , I_l)$, then for all strings $x \in \{0,1\}^{n}$, $P(x) = \prod_{i=1}^{l} I_{i}(x)$ where $I_i$'s are instructions over the monoid $\mathcal{M}$, $l$ is a polynomial in $n$ and product is the operation over monoid. Given an accepting set $\mathcal{A} \subseteq \mathcal{M}$, we say a program $P$ recognizes string $x$ if and only if $P(x) \in \mathcal{A}$. We say a language $L$ is recognized by a family of {\BWBP}, $\langle P_n \rangle$ if and only if $P_n$ recognizes exactly the set of all length $n$ strings in $L$. 

In a seminal work in 1986, Barrington gave the following characterization of $\NC^1$.

\begin{theorem}\label{thm:barring}
\cite{Barrington-NUDFA} A language $L$ is in $\NC^1$ if and only if $L$ is recognized by a family of {\BWBP} over some finite monoid.
\end{theorem}
In the following year Barrington and Th\'{e}rien extended their characterization to other subclasses in $\NC^1$. 
\begin{theorem}\label{thm:barring_therien}
\cite{Barrington-Therien} For a language $L$ we have,
\begin{enumerate}
\item $L$ is in ${\AC}^0$ if and only if $L$ is recognized by a family of {\BWBP} over an aperiodic finite monoid, 
\item \label{thm:ac0m} $L$ is in ${\AC}^0[p^{\alpha}]$ for a prime $p$ and constant $\alpha$ if and only if $L$ is recognized by a family of {\BWBP} over a solvable finite monoid in which all groups have order that divide power of $p$, and,
\item $L$ is in ${\ACC}^0$ if and only if $L$ is recognized by a family of {\BWBP} over a solvable finite monoid, 
\end{enumerate}
\end{theorem}

Part \ref{thm:ac0m} of Theorem \ref{thm:barring_therien} is not directly stated in \cite{Barrington-Therien} but can be derived from their proof as also claimed in \cite{Hansen2014}. We will use these results to show that {\BGGM} is not in $\AC^0[p^\alpha]$ where $p$ is prime and $\alpha \in \mathbb{N}$.

%
%
%
%

\section{{\BGGM} is as hard as {\Parity}}
\label{sec:bggm2parity}
In this section we will give an $\AC^0$ reduction from {\Parity} to {\BGGM}. 

Let $x = x_1x_2\ldots x_n \in \{0,1\}^*$ be an instance of {\Parity}. Define a function $f(x)$ as 
\[f(x) = 0 \bd (0x_10x_20\ldots 0x_n 0)0.\]
where $\bd$ is the bit-double function defined as $\bd(y_1y_2 \ldots y_n) = y_1y_1y_2y_2 \ldots y_ny_n$. Clearly, $f$ is an $\AC^0$ computable function. Note that $f(x)$ always has even length and we can visualize $f(x)$ as concatenation of pairs of $2$ bits. That is, the first pair contains the first and second bits of $f(x)$, second pair contains the third and fourth bits of $f(x)$ and so on. We will call these pairs as {\em constituent pairs} of $f(x)$. We note some properties of $f(x)$ that can easily be verified.

\begin{lemma}\label{lemma:pair}
For every string $x \in \{0,1\}^*$,
\begin{itemize}
\item $11$ cannot be a constituent pair of $f(x)$, and
\item in $f(x)$, a constituent pair $01$ is always succeeded by the constituent pair of $10$ and a constituent pair $10$ is always preceded by the constituent pair $01$.
\end{itemize}
\end{lemma}

Using $f(x)$ we construct a bipartite grid layered planar graph $G_x$, such that, $G_x$ has a perfect matching if and only if $x$ has even parity. Also, we will show that $G_x$ can be constructed from $f(x)$ in $\AC^0$. This will imply that {\Parity} reduces to {\BGGM}.

First we define graph blocks $G_{00}$, $G_{01}$ and $G_{10}$ corresponding to the three constituent pairs $00$, $01$ and $10$ respectively as shown in Figure \ref{fig:Blocks}. Note that $11$ cannot be a constituent pair hence we do not define a graph corresponding to it. These graph blocks will be the constituent elements of the graph $G_x$. 
\captionsetup[figure]{justification=centering}
\begin{figure}[!htb]
\captionsetup[subfigure]{justification=centering}
    \begin{subfigure}{.3\textwidth}
        \centering
\begin{tikzpicture}[scale=0.60]
\foreach \x in {0,1}
\foreach \y in {0,1,2,3,4,5}
    \filldraw[fill=red!10!white, draw=blue] (\x cm,\y cm)  circle (0.05cm);
\foreach \x in {0,1}
    \filldraw[fill=black!90!white, draw=black] (\x cm,0cm) circle (0.1cm);
\foreach \x in {0,1}
    \filldraw[fill=black!90!white, draw=black] (\x cm,5cm) circle (0.1cm);
\filldraw[fill=black!90!white, draw=black] (0cm,2cm) circle (0.1cm);
\filldraw[fill=black!90!white, draw=black] (0cm,3cm) circle (0.1cm);
\draw [black,thick] (0,2) -- (0,3);
\draw [black,thick] (0,0) -- (1,0);
\draw [black,thick] (0,5) -- (1,5);
\end{tikzpicture}
\caption{$G_{00}$}
\label{fig:G00}
    \end{subfigure}
    \begin{subfigure}{0.3\textwidth}
        \centering
\begin{tikzpicture}[scale=0.60]
\foreach \x in {0,1}
\foreach \y in {0,1,2,3,4,5}
    \filldraw[fill=red!10!white, draw=blue] (\x cm,\y cm)  circle (0.05cm);
\foreach \x in {0,1}
    \filldraw[fill=black!90!white, draw=black] (\x cm,1cm) circle (0.1cm);
\foreach \x in {0,1}
    \filldraw[fill=black!90!white, draw=black] (\x cm,4cm) circle (0.1cm);
\filldraw[fill=black!90!white, draw=black] (1cm,2cm) circle (0.1cm);
\filldraw[fill=black!90!white, draw=black] (1cm,3cm) circle (0.1cm);
\draw [black,thick] (0,1) -- (1,1);
\draw [black,thick] (0,4) -- (1,4);
\draw [black,thick] (1,2) -- (1,3);
\end{tikzpicture}
\caption{$G_{01}$}
\label{fig:G01}
    \end{subfigure}
    \begin{subfigure}{0.3\textwidth}
        \centering
\begin{tikzpicture}[scale=0.60]
\foreach \x in {0,1}
\foreach \y in {0,1,2,3,4,5}
    \filldraw[fill=red!10!white, draw=blue] (\x cm,\y cm)  circle (0.05cm);
\foreach \x in {0,1}
    \filldraw[fill=black!90!white, draw=black] (\x cm,0cm) circle (0.1cm);
\foreach \x in {0,1}
    \filldraw[fill=black!90!white, draw=black] (\x cm,5cm) circle (0.1cm);
\filldraw[fill=black!90!white, draw=black] (0cm,1cm) circle (0.1cm);
\filldraw[fill=black!90!white, draw=black] (0cm,4cm) circle (0.1cm);
\draw [black,thick] (0,0) -- (1,0);
\draw [black,thick] (0,0) -- (0,1);
\draw [black,thick] (0,5) -- (1,5);
\draw [black,thick] (0,5) -- (0,4);
\end{tikzpicture}
\caption{$G_{10}$}
\label{fig:G10}
\end{subfigure}
\caption{Different types of Graph Blocks}
\label{fig:Blocks}
\end{figure}
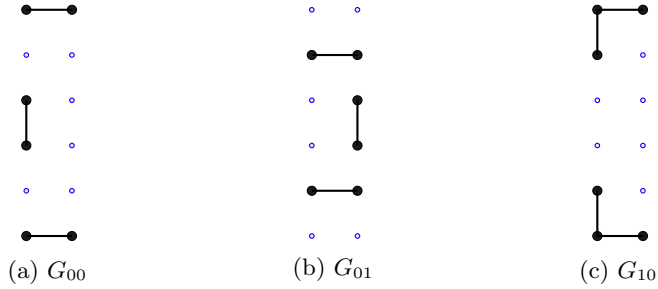

We also define an operator $\odot$ over these graphs which allows us to define larger graphs using these graph blocks. $\odot$ operator is defined in Figure \ref{fig:Rules}.

\begin{figure}[!htb]
\captionsetup[subfigure]{justification=centering}
    \centering
    \begin{subfigure}{.2\textwidth}
    \centering
    \begin{center}
    \begin{tikzpicture}[scale=0.60]
\foreach \x in {0,1}
\foreach \y in {0,1,2,3,4,5}
    \filldraw[fill=red!10!white, draw=blue] (\x cm,\y cm)  circle (0.05cm);
\foreach \x in {0,1}
    \filldraw[fill=black!90!white, draw=black] (\x cm,0cm) circle (0.1cm);
\foreach \x in {0,1}
    \filldraw[fill=black!90!white, draw=black] (\x cm,5cm) circle (0.1cm);
\filldraw[fill=black!90!white, draw=black] (0cm,2cm) circle (0.1cm);
\filldraw[fill=black!90!white, draw=black] (0cm,3cm) circle (0.1cm);
\draw [black,thick] (0,2) -- (0,3);
\draw [black,thick] (0,0) -- (1,0);
\draw [black,thick] (0,5) -- (1,5);
\foreach \x in {2,3}
\foreach \y in {0,1,2,3,4,5}
    \filldraw[fill=red!10!white, draw=blue] (\x cm,\y cm)  circle (0.05cm);
\foreach \x in {2,3}
    \filldraw[fill=black!90!white, draw=black] (\x cm,0cm) circle (0.1cm);
\foreach \x in {2,3}
    \filldraw[fill=black!90!white, draw=black] (\x cm,5cm) circle (0.1cm);
\filldraw[fill=black!90!white, draw=black] (2cm,2cm) circle (0.1cm);
\filldraw[fill=black!90!white, draw=black] (2cm,3cm) circle (0.1cm);
\draw [black,thick] (2,2) -- (2,3);
\draw [black,thick] (2,0) -- (3,0);
\draw [black,thick] (2,5) -- (3,5);
\draw [black,thick] (1,5) -- (2,5);
\draw [black,thick] (1,0) -- (2,0);
\end{tikzpicture}
\end{center}
\caption{$G_{00}\odot G_{00}$}
\label{fig:G00G00}
\end{subfigure}
\hfill
\begin{subfigure}{0.2\textwidth}
    \centering
\begin{center}
\begin{tikzpicture}[scale=0.60]
\foreach \x in {0,1}
\foreach \y in {0,1,2,3,4,5}
    \filldraw[fill=red!10!white, draw=blue] (\x cm,\y cm)  circle (0.05cm);
\foreach \x in {0,1}
    \filldraw[fill=black!90!white, draw=black] (\x cm,0cm) circle (0.1cm);
\foreach \x in {0,1}
    \filldraw[fill=black!90!white, draw=black] (\x cm,5cm) circle (0.1cm);
\filldraw[fill=black!90!white, draw=black] (0cm,2cm) circle (0.1cm);
\filldraw[fill=black!90!white, draw=black] (0cm,3cm) circle (0.1cm);
\draw [black,thick] (0,2) -- (0,3);
\draw [black,thick] (0,0) -- (1,0);
\draw [black,thick] (0,5) -- (1,5);
\foreach \x in {2,3}
\foreach \y in {0,1,2,3,4,5}
    \filldraw[fill=red!10!white, draw=blue] (\x cm,\y cm)  circle (0.05cm);
\foreach \x in {2,3}
    \filldraw[fill=black!90!white, draw=black] (\x cm,1cm) circle (0.1cm);
\foreach \x in {2,3}
    \filldraw[fill=black!90!white, draw=black] (\x cm,4cm) circle (0.1cm);
\filldraw[fill=black!90!white, draw=black] (3cm,2cm) circle (0.1cm);
\filldraw[fill=black!90!white, draw=black] (3cm,3cm) circle (0.1cm);
\draw [black,thick] (2,1) -- (3,1);
\draw [black,thick] (2,4) -- (3,4);
\draw [black,thick] (3,2) -- (3,3);
\draw [black,thick] (1,5) -- (2,4);
\draw [black,thick] (1,0) -- (2,1);
\end{tikzpicture}
\end{center}
\caption{$G_{00}\odot G_{01}$}
\label{fig:G00G01}
\end{subfigure}
\hfill
\begin{subfigure}{0.2\textwidth}
\centering
\begin{center}
\begin{tikzpicture}[scale=0.60]
\foreach \x in {0,1}
\foreach \y in {0,1,2,3,4,5}
    \filldraw[fill=red!10!white, draw=blue] (\x cm,\y cm)  circle (0.05cm);
\foreach \x in {0,1}
    \filldraw[fill=black!90!white, draw=black] (\x cm,0cm) circle (0.1cm);
\foreach \x in {0,1}
    \filldraw[fill=black!90!white, draw=black] (\x cm,5cm) circle (0.1cm);
\filldraw[fill=black!90!white, draw=black] (0cm,1cm) circle (0.1cm);
\filldraw[fill=black!90!white, draw=black] (0cm,4cm) circle (0.1cm);
\draw [black,thick] (0,0) -- (1,0);
\draw [black,thick] (0,0) -- (0,1);
\draw [black,thick] (0,5) -- (1,5);
\draw [black,thick] (0,5) -- (0,4);
\foreach \x in {2,3}
\foreach \y in {0,1,2,3,4,5}
    \filldraw[fill=red!10!white, draw=blue] (\x cm,\y cm)  circle (0.05cm);
\foreach \x in {2,3}
    \filldraw[fill=black!90!white, draw=black] (\x cm,0cm) circle (0.1cm);
\foreach \x in {2,3}
    \filldraw[fill=black!90!white, draw=black] (\x cm,5cm) circle (0.1cm);
\filldraw[fill=black!90!white, draw=black] (2cm,2cm) circle (0.1cm);
\filldraw[fill=black!90!white, draw=black] (2cm,3cm) circle (0.1cm);
\draw [black,thick] (2,2) -- (2,3);
\draw [black,thick] (2,0) -- (3,0);
\draw [black,thick] (2,5) -- (3,5);
\draw [black,thick] (1,0) -- (2,0);
\draw [black,thick] (1,5) -- (2,5);
\end{tikzpicture}
\end{center}
\caption{$G_{10}\odot G_{00}$}
\label{fig:G10G00}
\end{subfigure}
\hfill
\\
\vspace*{1cm}
\hfill
\begin{subfigure}{.2\textwidth}
\begin{center}
\begin{tikzpicture}[scale=0.60]
\foreach \x in {0,1}
\foreach \y in {0,1,2,3,4,5}
    \filldraw[fill=red!10!white, draw=blue] (\x cm,\y cm)  circle (0.05cm);
\foreach \x in {0,1}
    \filldraw[fill=black!90!white, draw=black] (\x cm,1cm) circle (0.1cm);
\foreach \x in {0,1}
    \filldraw[fill=black!90!white, draw=black] (\x cm,4cm) circle (0.1cm);
\filldraw[fill=black!90!white, draw=black] (1cm,2cm) circle (0.1cm);
\filldraw[fill=black!90!white, draw=black] (1cm,3cm) circle (0.1cm);
\draw [black,thick] (0,1) -- (1,1);
\draw [black,thick] (0,4) -- (1,4);
\draw [black,thick] (1,2) -- (1,3);
\foreach \x in {2,3}
\foreach \y in {0,1,2,3,4,5}
    \filldraw[fill=red!10!white, draw=blue] (\x cm,\y cm)  circle (0.05cm);
\foreach \x in {2,3}
    \filldraw[fill=black!90!white, draw=black] (\x cm,0cm) circle (0.1cm);
\foreach \x in {2,3}
    \filldraw[fill=black!90!white, draw=black] (\x cm,5cm) circle (0.1cm);
\filldraw[fill=black!90!white, draw=black] (2cm,1cm) circle (0.1cm);
\filldraw[fill=black!90!white, draw=black] (2cm,4cm) circle (0.1cm);
\draw [black,thick] (2,0) -- (3,0);
\draw [black,thick] (2,0) -- (2,1);
\draw [black,thick] (2,5) -- (3,5);
\draw [black,thick] (2,5) -- (2,4);
\draw [black,thick] (1,1) -- (2,1);
\draw [black,thick] (1,4) -- (2,4);
\end{tikzpicture}
\end{center}

\caption{$G_{01}\odot G_{10}$}
\label{fig:G01G10}
\end{subfigure}
\hfill
\begin{subfigure}{0.2\textwidth}
\centering
\begin{center}
\begin{tikzpicture}[scale=0.60]
\foreach \x in {0,1}
\foreach \y in {0,1,2,3,4,5}
    \filldraw[fill=red!10!white, draw=blue] (\x cm,\y cm)  circle (0.05cm);
\foreach \x in {0,1}
    \filldraw[fill=black!90!white, draw=black] (\x cm,0cm) circle (0.1cm);
\foreach \x in {0,1}
    \filldraw[fill=black!90!white, draw=black] (\x cm,5cm) circle (0.1cm);
\filldraw[fill=black!90!white, draw=black] (0cm,1cm) circle (0.1cm);
\filldraw[fill=black!90!white, draw=black] (0cm,4cm) circle (0.1cm);
\draw [black,thick] (0,0) -- (1,0);
\draw [black,thick] (0,0) -- (0,1);
\draw [black,thick] (0,5) -- (1,5);
\draw [black,thick] (0,5) -- (0,4);
\foreach \x in {2,3}
\foreach \y in {0,1,2,3,4,5}
    \filldraw[fill=red!10!white, draw=blue] (\x cm,\y cm)  circle (0.05cm);
\foreach \x in {2,3}
    \filldraw[fill=black!90!white, draw=black] (\x cm,1cm) circle (0.1cm);
\foreach \x in {2,3}
    \filldraw[fill=black!90!white, draw=black] (\x cm,4cm) circle (0.1cm);
\filldraw[fill=black!90!white, draw=black] (3cm,2cm) circle (0.1cm);
\filldraw[fill=black!90!white, draw=black] (3cm,3cm) circle (0.1cm);
\draw [black,thick] (2,1) -- (3,1);
\draw [black,thick] (2,4) -- (3,4);
\draw [black,thick] (3,2) -- (3,3);
\draw [black,thick] (1,0) -- (2,1);
\draw [black,thick] (1,5) -- (2,4);
\end{tikzpicture}
\end{center}

\caption{$G_{10}\odot G_{01}$}
\label{fig:G10G01}
\end{subfigure}
\hfill
\hfill
\caption{Joining different Blocks with $\odot$ operation}
\label{fig:Rules}
\end{figure}

We now complete the construction of $G_x$. Let $y=f(x)=y_1y_2\ldots y_m$. where $m$ is even. Then
\[G_x = G_{y_1y_2} \odot G_{y_3y_4} \odot \ldots \odot G_{y_{m-1} y_{m}}.\]
For example if $x=1101$, then $f(x) = 0bd(010100010)0=00 01 10 01 10 00 00 01 10 00$ and $G_x$ will be as shown in Figure \ref{fig:exampleG}.

\begin{figure}

\captionsetup[subfigure]{justification=centering}
    \centering
\begin{center}
\begin{tikzpicture}[scale=0.60]
\foreach \x in {0,1}
\foreach \y in {0,1,2,3,4,5}
    \filldraw[fill=red!10!white, draw=blue] (\x cm,\y cm)  circle (0.05cm);
\foreach \x in {0,1}
    \filldraw[fill=black!90!white, draw=black] (\x cm,0cm) circle (0.1cm);
\foreach \x in {0,1}
    \filldraw[fill=black!90!white, draw=black] (\x cm,5cm) circle (0.1cm);
\filldraw[fill=black!90!white, draw=black] (0cm,2cm) circle (0.1cm);
\filldraw[fill=black!90!white, draw=black] (0cm,3cm) circle (0.1cm);
\draw [black,thick] (0,2) -- (0,3);
\draw [black,thick] (0,0) -- (1,0);
\draw [black,thick] (0,5) -- (1,5);
\foreach \x in {2,3}
\foreach \y in {0,1,2,3,4,5}
    \filldraw[fill=red!10!white, draw=blue] (\x cm,\y cm)  circle (0.05cm);
\foreach \x in {2,3}
    \filldraw[fill=black!90!white, draw=black] (\x cm,1cm) circle (0.1cm);
\foreach \x in {2,3}
    \filldraw[fill=black!90!white, draw=black] (\x cm,4cm) circle (0.1cm);
\filldraw[fill=black!90!white, draw=black] (3cm,2cm) circle (0.1cm);
\filldraw[fill=black!90!white, draw=black] (3cm,3cm) circle (0.1cm);
\draw [black,thick] (2,1) -- (3,1);
\draw [black,thick] (2,4) -- (3,4);
\draw [black,thick] (3,2) -- (3,3);
\draw [black,thick] (1,5) -- (2,4);
\draw [black,thick] (1,0) -- (2,1);

\foreach \x in {4,5}
\foreach \y in {0,1,2,3,4,5}
    \filldraw[fill=red!10!white, draw=blue] (\x cm,\y cm)  circle (0.05cm);
\foreach \x in {4,5}
    \filldraw[fill=black!90!white, draw=black] (\x cm,0cm) circle (0.1cm);
\foreach \x in {4,5}
    \filldraw[fill=black!90!white, draw=black] (\x cm,5cm) circle (0.1cm);
\filldraw[fill=black!90!white, draw=black] (4cm,1cm) circle (0.1cm);
\filldraw[fill=black!90!white, draw=black] (4cm,4cm) circle (0.1cm);
\draw [black,thick] (4,0) -- (5,0);
\draw [black,thick] (4,0) -- (4,1);
\draw [black,thick] (4,5) -- (5,5);
\draw [black,thick] (4,5) -- (4,4);
\draw [black,thick] (3,1) -- (4,1);
\draw [black,thick] (3,4) -- (4,4);

\foreach \x in {6,7}
\foreach \y in {0,1,2,3,4,5}
    \filldraw[fill=red!10!white, draw=blue] (\x cm,\y cm)  circle (0.05cm);
\foreach \x in {6,7}
    \filldraw[fill=black!90!white, draw=black] (\x cm,1cm) circle (0.1cm);
\foreach \x in {6,7}
    \filldraw[fill=black!90!white, draw=black] (\x cm,4cm) circle (0.1cm);
\filldraw[fill=black!90!white, draw=black] (7cm,2cm) circle (0.1cm);
\filldraw[fill=black!90!white, draw=black] (7cm,3cm) circle (0.1cm);
\draw [black,thick] (6,1) -- (7,1);
\draw [black,thick] (6,4) -- (7,4);
\draw [black,thick] (7,2) -- (7,3);
\draw [black,thick] (5,5) -- (6,4);
\draw [black,thick] (5,0) -- (6,1);

\foreach \x in {8,9}
\foreach \y in {0,1,2,3,4,5}
    \filldraw[fill=red!10!white, draw=blue] (\x cm,\y cm)  circle (0.05cm);
\foreach \x in {8,9}
    \filldraw[fill=black!90!white, draw=black] (\x cm,0cm) circle (0.1cm);
\foreach \x in {8,9}
    \filldraw[fill=black!90!white, draw=black] (\x cm,5cm) circle (0.1cm);
\filldraw[fill=black!90!white, draw=black] (8cm,1cm) circle (0.1cm);
\filldraw[fill=black!90!white, draw=black] (8cm,4cm) circle (0.1cm);
\draw [black,thick] (8,0) -- (9,0);
\draw [black,thick] (8,0) -- (8,1);
\draw [black,thick] (8,5) -- (9,5);
\draw [black,thick] (8,5) -- (8,4);
\draw [black,thick] (7,1) -- (8,1);
\draw [black,thick] (7,4) -- (8,4);

\foreach \x in {10,11}
\foreach \y in {0,1,2,3,4,5}
    \filldraw[fill=red!10!white, draw=blue] (\x cm,\y cm)  circle (0.05cm);
\foreach \x in {10,11}
    \filldraw[fill=black!90!white, draw=black] (\x cm,0cm) circle (0.1cm);
\foreach \x in {10,11}
    \filldraw[fill=black!90!white, draw=black] (\x cm,5cm) circle (0.1cm);
\filldraw[fill=black!90!white, draw=black] (10cm,2cm) circle (0.1cm);
\filldraw[fill=black!90!white, draw=black] (10cm,3cm) circle (0.1cm);
\draw [black,thick] (10,2) -- (10,3);
\draw [black,thick] (10,0) -- (11,0);
\draw [black,thick] (10,5) -- (11,5);
\draw [black,thick] (9,0) -- (10,0);
\draw [black,thick] (9,5) -- (10,5);

\foreach \x in {12,13}
\foreach \y in {0,1,2,3,4,5}
    \filldraw[fill=red!10!white, draw=blue] (\x cm,\y cm)  circle (0.05cm);
\foreach \x in {12,13}
    \filldraw[fill=black!90!white, draw=black] (\x cm,0cm) circle (0.1cm);
\foreach \x in {12,13}
    \filldraw[fill=black!90!white, draw=black] (\x cm,5cm) circle (0.1cm);
\filldraw[fill=black!90!white, draw=black] (12cm,2cm) circle (0.1cm);
\filldraw[fill=black!90!white, draw=black] (12cm,3cm) circle (0.1cm);
\draw [black,thick] (12,2) -- (12,3);
\draw [black,thick] (12,0) -- (13,0);
\draw [black,thick] (12,5) -- (13,5);
\draw [black,thick] (11,0) -- (12,0);
\draw [black,thick] (11,5) -- (12,5);

\foreach \x in {14,15}
\foreach \y in {0,1,2,3,4,5}
    \filldraw[fill=red!10!white, draw=blue] (\x cm,\y cm)  circle (0.05cm);
\foreach \x in {14,15}
    \filldraw[fill=black!90!white, draw=black] (\x cm,1cm) circle (0.1cm);
\foreach \x in {14,15}
    \filldraw[fill=black!90!white, draw=black] (\x cm,4cm) circle (0.1cm);
\filldraw[fill=black!90!white, draw=black] (15cm,2cm) circle (0.1cm);
\filldraw[fill=black!90!white, draw=black] (15cm,3cm) circle (0.1cm);
\draw [black,thick] (14,1) -- (15,1);
\draw [black,thick] (14,4) -- (15,4);
\draw [black,thick] (15,2) -- (15,3);
\draw [black,thick] (13,5) -- (14,4);
\draw [black,thick] (13,0) -- (14,1);

\foreach \x in {16,17}
\foreach \y in {0,1,2,3,4,5}
    \filldraw[fill=red!10!white, draw=blue] (\x cm,\y cm)  circle (0.05cm);
\foreach \x in {16,17}
    \filldraw[fill=black!90!white, draw=black] (\x cm,0cm) circle (0.1cm);
\foreach \x in {16,17}
    \filldraw[fill=black!90!white, draw=black] (\x cm,5cm) circle (0.1cm);
\filldraw[fill=black!90!white, draw=black] (16cm,1cm) circle (0.1cm);
\filldraw[fill=black!90!white, draw=black] (16cm,4cm) circle (0.1cm);
\draw [black,thick] (16,0) -- (17,0);
\draw [black,thick] (16,0) -- (16,1);
\draw [black,thick] (16,5) -- (17,5);
\draw [black,thick] (16,5) -- (16,4);
\draw [black,thick] (15,1) -- (16,1);
\draw [black,thick] (15,4) -- (16,4);

\foreach \x in {18,19}
\foreach \y in {0,1,2,3,4,5}
    \filldraw[fill=red!10!white, draw=blue] (\x cm,\y cm)  circle (0.05cm);
\foreach \x in {18,19}
    \filldraw[fill=black!90!white, draw=black] (\x cm,0cm) circle (0.1cm);
\foreach \x in {18,19}
    \filldraw[fill=black!90!white, draw=black] (\x cm,5cm) circle (0.1cm);
\filldraw[fill=black!90!white, draw=black] (18cm,2cm) circle (0.1cm);
\filldraw[fill=black!90!white, draw=black] (18cm,3cm) circle (0.1cm);
\draw [black,thick] (18,2) -- (18,3);
\draw [black,thick] (18,0) -- (19,0);
\draw [black,thick] (18,5) -- (19,5);
\draw [black,thick] (17,0) -- (18,0);
\draw [black,thick] (17,5) -- (18,5);

\end{tikzpicture}
\end{center}
\caption{Graph $G_x$ corresponding to the string $x = 1101$}
\label{fig:exampleG}
\end{figure}

\begin{lemma} \label{lem:bggm}
For every string $x \in \{ 0,1 \}^*$, $G_x$ is a bipartite grid layered planar graph. Also, each connected component of $G_x$ is either a single edge or a path that extends from the first block to the last block of $G_x$.
\end{lemma}
\begin{proof}
By definition, each of the three graph blocks is grid layered planar graphs. Moreover, the operator $\odot$ connects adjacent blocks by preserving planarity and the overall grid structure. Hence $G_x$ is a grid layered planar graph.

To show that each connected component of $G_x$ is a path we will use induction on the number constituent pairs of $y=f(x)$. For the base case note that if $y$ has only one constituent pair then it must be $00$ and $G_{00}$ contains only paths of even length (number of vertices). Now consider a graph $G_p$ corresponding to the first $m-1$ constituent pairs of $y$. Assume that every connected component in $G_p$ is a path. Let $G_{p'}$ be the graph corresponding to the first $m$ constituent pairs of $y$. If the last block of $G_p$ is $G_{00}$ then by Lemma \ref{lemma:pair}, the next block can either be $G_{00}$ or $G_{01}$. By Figure \ref{fig:G00G00} and \ref{fig:G00G01} we have that $G_{p'}$ will only be extending the paths of $G_p$ in addition to two isolated edges. So every connected component in $G_{p'}$ will be a path as well. Similarly if the last block of $G_p$ is $G_{01}$ then again by Lemma \ref{lemma:pair}, the next block will be $G_{10}$ and by Figure \ref{fig:G01G10} we have that every connected component in $G_{p'}$ will be a path as well. Finally if the last block of $G_p$ is $G_{10}$ then the next block can either be $G_{00}$ or $G_{01}$ and by Figure \ref{fig:G10G00} and \ref{fig:G10G01} we have that every connected component in $G_{p'}$ will be a path as well. Also, note that each path in $G_{p'}$ extends from the first block to the last one or is of length one.

This also shows that $G_x$ is bipartite since it does not have any cycles.
\end{proof}

\begin{lemma} \label{lem:equiv}
For every string $x \in \{ 0,1 \}^*$, $G_x$ has a perfect matching if and only if $x$ has even parity.
\end{lemma}
\begin{proof}
We claim that $G_x$ has a perfect matching if and only if it has an even number of $G_{10}$ blocks. Since the number of $G_{10}$ blocks in $G_x$ is same as the number of ones in $x$, this will complete the proof. To prove our claim we again use induction on the number constituent pairs of $y$. 

For the base case note that $G_{00}$ has a perfect matching using all its three edges. Assume we have a graph $G_p$ corresponding to the first $m-1$ constituent pairs of $y$ such that $G_p$ has a perfect matching if and only if $G_p$ has an even number of $G_{10}$ blocks. Now suppose we are extending the graph $G_p$ by one graph block to get the graph $G_{p'}$. We divide this into two cases.
\begin{description}
\item [Case 1: $G_{p'} = G_p \odot G_{00}$ or $G_{p'} = G_p \odot G_{01}$.] In this case $G_p$ and $G_{p'}$ have the same number of $G_{10}$ blocks. By construction two paths in $G_{p'}$ get extended by two vertices while others remain the same. Also, an additional new edge is introduced whose endpoints are matched with each other (see Figure \ref{fig:Rules}). Thus if $G_p$ has a perfect matching then $G_{p'}$ will also have a perfect matching where the two new vertices by which the paths get extended, are matched with each other. If $G_p$ does not have a perfect matching then at least one of its paths has an odd number of vertices. Extending this path by two more vertices preserves its parity and hence $G_{p'}$ will also not have a perfect matching.

\item [Case 2: $G_{p'} = G_p \odot G_{10}$.] In this case $G_{p'}$ has an extra $G_{10}$ block from $G_{p}$. Two paths in $G_{p'}$ get extended by three vertices while others remain the same (see Figure \ref{fig:G01G10}). Also, note that each path is symmetric about a horizontal axis passing through the centre. Therefore both paths of length more than one has the same length. If $G_p$ has a perfect matching then the new graph $G_{p'}$ does not have perfect matching as the number of vertices in the two paths become odd and hence cannot be matched. On the other hand, if $G_p$ does not have a perfect matching then both the long paths have an odd number of vertices. Hence in $G_{p'}$ these paths will have an even number of vertices and hence a perfect matching exists in $G_{p'}$.
\end{description}
\end{proof}

\begin{lemma} \label{lem:comp}
For every string $x \in \{ 0,1 \}^*$, $G_x$ and its planar embedding is computable in $\AC^0$. 
\end{lemma}
\begin{proof}
It is easy to see that $G_x$ can be computed using $\AC^0$ circuits when $f(x)$ is given as input as each vertex depends on at most two bits of $f(x)$ and each edge depends on at most four bits of $f(x)$. Also the embedding is $\AC^0$-computable as we can compute whether $(a,b)$ is a vertex using just two bits of $f(x)$ and whether $(a,b)$ and $(c,d)$ have edge between them using just four bits of $f(x)$. Next we show that $G_x$ and its planar embedding is $\AC^0$-computable even when $x$ is given as input. For this note that each bit of $f(x)$ is either $0$ or a copy of some bit in $x$. In circuit taking $y$ as input, we can hardcode some bits to $0$ and pass input from bits of $x$ wherever copy bit of some bit in $x$ are used. Thus, we get a circuit which computes $G_x$ and its planar embedding using $x$ as input.  
\end{proof}

\begin{theorem} \label{thm:red}
$\Parity$ reduces to $\BGGM$ in $\AC^0$. 
\end{theorem}
Theorem \ref{thm:red} follows from Lemmas \ref{lem:bggm}, \ref{lem:equiv} and \ref{lem:comp}. Razborov and Smolensky had independently shown the following result.
\begin{theorem} \label{thm:razsmo}
\cite{Raz89, Smo87}
Let $p$ and $q$ be two distinct prime numbers and $\alpha \in \mathbb{N}$. $\Mod_q \notin {\AC}^0[p^\alpha]$.
\end{theorem}

Now by combining the result of Razborov and Smolensky and applying Theorem \ref{thm:red} we have the following result.
\begin{theorem} \label{thm:final1}
$\BGGM$ is not in $\AC^0[p^{\alpha}]$ for every odd prime $p$ and $\alpha \in \mathbb{N}$.
\end{theorem}

\section{Lower Bounds for {\BGGM}}
\label{sec:bggmlowerbound}

Now we will show that $\BGGM$ is not in $\AC^{0}[2^{\alpha}]$ as well. For this, we will first describe reduction given in \cite{Hansen2014} that shows $\BGGM$ is in $\ACC^0$. 

\subsection{Reduction of BGGM to Monoid Word Problem}
\label{sec:monoid}
For two relations $R, S \subseteq \mathcal{A} \times \mathcal{B}$. We define $RS = \{(x,y) \mid \exists z \text{ such that } (x,z) \in R \text{ and } (z,y) \in S \}$. For any grid layered planar graph $G$, we will first define the monoid element corresponding to $G$. We will denote this monoid element by $G^{\mathcal{M}}$. For any $G$, $G^{\mathcal{M}} = (X,Y,R)$ where $X$ is set of vertices in leftmost layer of $G$, $Y$ is set of vertices in rightmost layer of $G$ and $R \subseteq 2^{X} \times 2^{Y}$ is a relation. Let $X' \subseteq X$ and $Y' \subseteq Y$ then $(X',Y') \in R$ if and only if $G \setminus ( \overline{X'} \cup Y')$ has a perfect matching.  In other words, there is a matching in $G$ which matches every vertex except those in $\overline{X'}$ and $Y'$.
We now define the monoid as the set
\[\mathcal{M} = \{ G^{\mathcal{M}}|\text{  }G\text{ is a bipartite grid layered planar graph} \} \cup \{0,1\}\] 
where $1$ is the identity element of the monoid and the operation of the monoid (denoted as $*$) is defined as
\[(W,X,R)*(Y,Z,S)  = 
\begin{cases}
     (W,Z,R S) &\text{if } X = Y\\
    0 & \text{if } X \ne Y
\end{cases}
\]
and for all $M \in \mathcal{M}$, $0*M = M*0 = 0$. For the remaining part of Section \ref{sec:bggmlowerbound} we will refer to this monoid as $\mathcal{M}$.

Next, we define a concatenation operation on grid layered planar graphs. Let $G_1$ and $G_2$ be two grid layered planar graphs having same width $w$ and lengths $l_1$ and $l_2$ respectively. We define $G_1 \cdot G_2$ to be the grid layered planar graph having width $w$ and length $l_1 + l_2$, obtained by identifying the vertices in the rightmost column of $G_1$ and the leftmost column of $G_2$. Here we assume that there are no vertical edges present in the leftmost or rightmost column of a grid layered planar graph. This can be assumed without loss of generality because given any grid layered planar graph we can convert it to a grid layered planar graph having the above property by adding additional columns to the left and right, and adding edges appropriately such that it preserves matching. Then we have the property that $(G_1 \cdot G_2)^{\mathcal{M}} = {G_1}^{\mathcal{M}}*{G_2}^{\mathcal{M}}$. 

\subsection{$\BGGM$ is not in $\AC^0[2^\alpha]$}
In this section we show that {\BGGM} is also not in $\AC^0[2^\alpha]$. First we will show an algebraic property of the monoid $\mathcal{M}$ defined in Section \ref{sec:monoid}.
\begin{lemma}\label{lemma:GroupP}
There exists a cyclic group $\mathcal{G}$ in $\mathcal{M}$ of order $p$ where $p$ is an odd prime.
\end{lemma}
\begin{proof}
Since by Theorem \ref{thm:red} $\BGGM$ is not in  $\AC^0$, hence $\mathcal{M}$ is not an aperiodic monoid by the characterisation provided in \cite{Barrington-Therien}. Thus a group $\mathcal{G}_{n} \subseteq \mathcal{M}$ and $\mid \mathcal{G}_{n} \mid = n >1$ exists. Hansen et al. showed that every group contained in $\mathcal{M}$ will have odd order \cite{Hansen2014}. Thus, some odd prime $p | ord(\mathcal{G}_{n})$. Thus there exist a cyclic subgroup of $\mathcal{M}_{n}$, $\mathcal{G}$, of order $p$.
\end{proof}

\begin{lemma}\label{lemma:SameLen}
Let $\mathcal{G}$ be a cyclic group of order $p$ in $\mathcal{M}$ and $e$ is the identity of $\mathcal{G}$. Then for some generator of $\mathcal{G}$, say $x$, there exists grid layered planar graphs $A$ and $B$ such that $A^{\mathcal{M}} = x$, $B^{\mathcal{M}} = e$ and $A$ and $B$ have same length.
\end{lemma}
\begin{proof}
We have $x, e \in \mathcal{M}$ such that $x \neq e$. Moreover it is easy to see that the elements $0$ and $1$ of  $\mathcal{M}$ are not contained in $\mathcal{G}$. Hence there exists grid layered planar graphs $A'$ and $B'$ such that $A^{\mathcal{M}} = x$ and $B^{\mathcal{M}} = e$. If $A$ and $B$ have same length we are done. Otherwise using $A$ and $B$ we will give the construction of grid layered planar graphs $A'$ and $B'$  such that lengths of $A'$ and $B'$ are same, $A'^{\mathcal{M}} = y$ and $B'^{\mathcal{M}} = e$ where $y$ is a generator of $\mathcal{G}$. Consider the following cases:
\begin{description}
\item [Case 1: Difference between the lengths of $A$ and $B$ is even.] Without loss of generality assume $A$ has smaller length. We construct $A^{\prime}$ by adding an even number of columns to the right hand side of $A$ such that $A'$ and $B$ have the same number of columns. Now we add horizontal paths of even length from each vertex in the rightmost column of $A$ to its corresponding vertex in the rightmost column of $A'$. Since we have added paths of even length, therefore there there is a one to one correspondence between matchings in $A$ and $A'$. Hence $A'^{\mathcal{M}} = A^{\mathcal{M}}$.
\item [Case 2: $A$ has odd length and $B$ has even length.] As $B^{\mathcal{M}} =e$ we have $(B \cdot B)^{\mathcal{M}}=e^2=e$. Also $B \cdot B$ will have odd length. Hence it reduces to Case 1.
\item [Case 3: $A$ has even length and $B$ has odd length.] Note that if $A^{\mathcal{M}}=x$ has order $p$ then $(A \cdot A)^{\mathcal{M}} = x^2$ also have order $p$. Now $A \cdot A$ has odd length. Hence it reduces to Case 1 again.
\end{description}
\end{proof}
Consider the grid layered planar graphs $A$ and $B$ as obtained by Lemma \ref{lemma:SameLen}. Using them we define a function $h$ from the set of all binary strings to the set of all grid layered planar graphs. $h$ is defined recursively as follows:
\begin{align*}
h(\epsilon) &= B \\
h(y0) &= h(y) \cdot B \\
h(y1) &= h(y) \cdot A 
\end{align*}
Note that $h(z)$ is essentially the grid layered planar graph obtained by concatenating copies of $A$ and $B$ for every $1$ and $0$ in the string $z$ respectively, together with an extra $B$ for $\epsilon$ at the leftmost end.  
\begin{lemma}\label{lemma:AC0Comp}
$h$ is $\AC^0$-computable.
\end{lemma}
\begin{proof}
By Lemma \ref{lemma:SameLen}, for every positive integer $n$ there exists grid layered planar graphs $A$ and $B$ such that $A^{\mathcal{M}} = x$, $B^{\mathcal{M}} = e$ and $A$ and $B$ have the same length (say $m$) and same width (say $w$). We assume that $A$ and $B$ are hardcoded into the $\AC^0$ circuit say $C_n$. Now given a string $z \in \{0,1\}^n$, for every bit $0$ or $1$ of $z$, the circuit $C_n$ outputs the corresponding graph $B$ or $A$ respectively in the order of the input bits. Additionally $C_n$ also outputs a copy of the graph $B$ at the beginning. Hence the output graph will have width $w$ and length $m + n(m-1)$. Here we will crucially use the fact that $A$ and $B$ have same the length, since otherwise the length of the output graph would have been variable.
\end{proof}

\begin{lemma}\label{lemma:equiv}
Let $\mathcal{G}$ be the group as obtained in Lemma \ref{lemma:GroupP} and let $e$ be the identity element in $\mathcal{G}$. For all strings $z \in \{0,1\}^*$, $z \in \Mod_p$ if and only if $h(z)^{\mathcal{M}} \neq e$.
\end{lemma}
\begin{proof}
Let $z = z_1z_2 \ldots z_n$ such that $z_i \in \{0,1\}$. Then,
\begin{align*}
h(z)^{\mathcal{M}} & =  h(z_1z_2 \ldots z_n)^{\mathcal{M}}\\
& =  (B \cdot X_1 \cdot X_2 \cdot \ldots \cdot X_n)^{\mathcal{M}} \textrm{, such that } X_i = A \textrm{ if } z_i = 1 \textrm{ and } X_i = B \textrm{ if } z_i = 0\\
& =  B^{\mathcal{M}}* X_1^{\mathcal{M}} * X_2^{\mathcal{M}}* \ldots * X_n^{\mathcal{M}}\\
& = x^t \text{, where $t$ is the number of 1's in $z$}
\end{align*}

Now by Lemma \ref{lemma:GroupP} we have $x^t = e$ if and only if $t \equiv 0 \mod p$. Hence $z \in \Mod_p$ if and only if $h(z)^{\mathcal{M}} \neq e$.
\end{proof}

\begin{lemma}\label{lemma:Hard}
Consider the language $L = \{ h(z) \mid h(z)^{\mathcal{M}} = e$ where $z\in \{0,1\}^{*} \}$. Then $L$ reduces to  $\BGGM$ in $\AC^0$.
\end{lemma}
\begin{proof}
We know that $h(z)$ is a grid layered planar graph having fixed width, say $w$. Let $e=(X,Y,R)$ and $h(z)^{\mathcal{M}}=(X_0,Y_0,R_0)$. Assume $e= (X, Y, R)$ are hardcoded in the circuit. We can easily check if $X_0 = X$ and $Y_0 = Y$ using an $\AC^0$ circuit. 

For checking whether $R_0=R$, we will create $2^w$ instances of $\BGGM$ and infer $R_0$ from their output.
For each $X' \subseteq X$ and $Y' \subseteq Y$, we create a graph $G_{X'Y'}$. $G_{X'Y'}$ is the graph $h(z)$ with some additional vertices and edges. For each $v \in \overline{X'}$ we add a vertex $l_v$ to the left of $v$ and the edge $\{v,l_v\}$. Similarly, For each $v \in Y'$ we add a vertex $r_v$ to the right of $v$ and the edge $\{v,r_v\}$. 
Note that $(X',Y') \in R_0$ if and only if $G_{X'Y'}$ have a perfect matching. Thus, we can compute $R_0$. Now $R_0$ if and only if for all $X' \subseteq X$ and $Y' \subseteq Y$ we have $G_{X'Y'}$ contains a perfect matching. So our output graph is essentially a union of all such graphs $G_{X'Y'}$. Again this can be constructed easily in $\AC^0$. 
\end{proof}

\begin{theorem} \label{thm:ModPred}
$\overline{\Mod_p}$ reduces to $\BGGM$ in $\AC^0$. 
\end{theorem}
\begin{proof}
Given $z \in \{0,1\}^n$ we first compute $h(z)$ using Lemma \ref{lemma:AC0Comp} and then output the grid layered planar graph (say $G_z$) as obtained by the reduction of Lemma \ref{lemma:Hard}. By Lemma \ref{lemma:equiv} and \ref{lemma:Hard} it follows that $z \notin \Mod_p$ if and only $G_z$ has a perfect matching.
\end{proof}

\begin{theorem} \label{thm:final2}
$\BGGM$ is not in $\AC^0[2^\alpha]$ for $\alpha \in \mathbb{N}$.
\end{theorem}
\begin{proof}
If {\BGGM} is in $\AC^0[2^\alpha]$ then by combining this circuit with the reduction in Theorem \ref{thm:ModPred}  and appending a $\NOT$ gate at the top we would get an $\AC^0[2^\alpha]$ circuit for $\Mod_p$. This would contradict Razborov and Smolensky's result (Theorem \ref{thm:razsmo}) since $p$ is an odd prime.
\end{proof}

%
Finally combining Theorems \ref{thm:final1} and \ref{thm:final2} we get the following theorem.

\begin{theorem}
$\BGGM$ is not in $\AC^0[p^\alpha]$ for every prime $p$ and $\alpha \in \mathbb{N}$. 
\end{theorem}

\section{Application of our results}
\label{sec:appl}

\subsection{Circuit Lower Bounds for Series-Parallel Graphs}
For series-parallel graphs, it is known that bipartite matching is in $\NC^1$ given transitive closure of tree decomposition as input as well \cite{Tantau-LogSpace}. We show a better lower bound for this problem using our reduction of $\Parity$ to $\BGGM$.

In reduction from $\Parity$ to $\BGGM$ the final graph we get is also a series-parallel graph (each connected component is a path or a single edge). Now the challenge is to construct the transitive closure of the tree decomposition(more precisely term representation of tree decomposition) for it using $\AC^0$ circuits. For this, we will use the following theorem mentioned in Hansen et al. \cite{Hansen2014}.

\begin{theorem}\label{thm:CutwidthToTreewidth}\cite{Hansen2014} 
Given as input a linear arrangement of bounded cutwidth $k$ for some graph $G$, a tree decomposition of width $k$ for graph $G$ in term representation can be constructed by an $\AC^0$ circuit.
\end{theorem}

\begin{lemma} \label{lem:CutwidthConstruction}
For every $x \in \{0,1\}^*$, a linear arrangement of bounded cutwidth for $G_x$ can be created by an $\AC^0$ circuit.
\end{lemma}
\begin{proof}
We linearize $G_{00}$, $G_{01}$ and $G_{10}$. For $G_{00}$ it is shown in figure \ref{fig:BlockG00SP}. Same can be done for $G_{01}$ and $G_{10}$. This order can be hardcoded in $\AC^0$ circuit. We keep edges same. Clearly to check edge between $v_i$ and $v_j$, we will not need more than $4$ bits. It can be shown that such a linear arrangement can be computed in ${\AC}^0$ using arguments similar to those given in Section \ref{sec:bggm2parity}.
\end{proof}
\begin{figure}[!htb]
\captionsetup[figure]{justification=centering}
    \centering
\begin{subfigure}[c]{.4\textwidth}
\begin{center}
\begin{tikzpicture}[scale=0.60]
\foreach \x in {0,1}
\foreach \y in {0,1,2,3,4,5}
    \filldraw[fill=white!10!white, draw=white] (\x cm,\y cm)  circle (0.05cm);

\foreach \x in {0,1}
\foreach \y in {0,1,2,3,4,5}
    \filldraw[fill=red!10!white, draw=blue] (\x cm,\y cm)  circle (0.05cm);

\filldraw[fill=black!90!white, draw=black] (0cm,0cm) circle (0.1cm) node[anchor=east] {A};
\filldraw[fill=black!90!white, draw=black] (1cm,0cm) circle (0.1cm) node[anchor=west] {B};

\filldraw[fill=black!90!white, draw=black] (0cm,5cm) circle (0.1cm) node[anchor=east] {C};
\filldraw[fill=black!90!white, draw=black] (1cm,5cm) circle (0.1cm) node[anchor=west] {D};

\filldraw[fill=black!90!white, draw=black] (0cm,2cm) circle (0.1cm) node[anchor=east] {E};
\filldraw[fill=black!90!white, draw=black] (0cm,3cm) circle (0.1cm) node[anchor=east] {F};
\draw [black,thick] (0,2) -- (0,3);
\draw [black,thick] (0,0) -- (1,0);
\draw [black,thick] (0,5) -- (1,5);
\end{tikzpicture}
\end{center}
\caption{$G_{00}$}
\label{fig:G00SP}
\end{subfigure}
\begin{subfigure}[c]{.4\textwidth}
\centering
\begin{center}
\begin{tikzpicture}[scale=0.60]
\foreach \x in {0,1,2,3,4,5}
\foreach \y in {0,1,2,3,4,5}
    \filldraw[fill=white!10!white, draw=white] (\x cm,\y cm)  circle (0.05cm);

\foreach \x in {0,1,2,3,4,5}
\foreach \y in {3}
    \filldraw[fill=red!10!white, draw=blue] (\x cm,\y cm)  circle (0.05cm);

\filldraw[fill=black!90!white, draw=black] (0cm,3cm) circle (0.1cm) node[anchor=east] {A};
\filldraw[fill=black!90!white, draw=black] (4cm,3cm) circle (0.1cm) node[anchor=west] {B};

\filldraw[fill=black!90!white, draw=black] (1cm,3cm) circle (0.1cm) node[anchor=east] {C};
\filldraw[fill=black!90!white, draw=black] (5cm,3cm) circle (0.1cm) node[anchor=west] {D};

\filldraw[fill=black!90!white, draw=black] (2cm,3cm) circle (0.1cm) node[anchor=east] {E};
\filldraw[fill=black!90!white, draw=black] (3cm,3cm) circle (0.1cm) node[anchor=west] {F};
\draw [black,thick] (0,3) to [bend right] (4,3);
\draw [black,thick] (1,3) to [bend left] (5,3);
\draw [black,thick] (2,3) -- (3,3);
\end{tikzpicture}
\end{center}
\caption{${G^{\prime}}_{00}$}
\label{fig:G00LinearSP}
\end{subfigure}

\caption{Linearizing $G_{00}$}
\label{fig:BlockG00SP}
\end{figure}

\begin{theorem} \label{thm:SPred}
$\Parity$ reduces to series-parallel graph matching with given tree decomposition in term representation.
\end{theorem}
\begin{proof}
Follows from Theorem \ref{thm:CutwidthToTreewidth} and Lemma \ref{lem:CutwidthConstruction}.
\end{proof}

\begin{theorem} \label{thm:SPLowerbound}
Series-parallel graph matching with given tree decomposition in term representation is not in $\AC^0[p^\alpha]$ where $p$ is odd prime and $\alpha \in \mathbb{N}$.
\end{theorem}
\begin{proof}
Follows from Theorems \ref{thm:SPred} and \ref{thm:razsmo}.
\end{proof}

\subsection{Future Work}
For $m\in \mathbb{N}$ and not a power of prime we do not know whether $\AC^0[m] = \NP$. To extend our results for all $m\in \mathbb{N}$ or to show that {\BGGM} lies in $\AC^0[m]$ for some $m \in \mathbb{N}$, algebraic characterization of these classes are needed. Algebraic theory of subclasses of ${\NC}^1$ developed by \cite{Barrington-Therien} does not provide any such characterization. This is the biggest hurdle in extending our approach to $\AC^0[m]$.

Our result improves the lower bound for perfect matching in series-parallel graphs but the bound are not tight. For example, we do not know if perfect matching in series-parallel graphs is in $\AC^0[2]$ or $\ACC^0$.

\section*{Acknowledgment}

The authors would like to thank Sarthak Garg for several sessions of helpful discussions. The second author would like to acknowledge Navid Talebanfard for introducing him to the problem. He would also like to thank Jayalal Sarma for inviting him to the PhD defence of Balagopal Komarath where he gained more insight into the area in general. 



\bibliography{references}


\end{document}